\documentclass[10pt,letterpaper]{article}
\usepackage[top=0.85in,footskip=0.75in,marginparwidth=2in]{geometry}

\usepackage[utf8]{inputenc}

\usepackage{cite}

\usepackage{nameref,hyperref}

\usepackage{microtype}
\DisableLigatures[f]{encoding = *, family = * }

\raggedright
\usepackage{changepage}

\usepackage[aboveskip=1pt,labelfont=bf,labelsep=period,singlelinecheck=off]{caption}

\makeatletter
\renewcommand{\@biblabel}[1]{\quad#1.}
\makeatother

\usepackage{lastpage,fancyhdr,graphicx}
\usepackage{epstopdf}
\fancyheadoffset[L]{2.25in}
\fancyfootoffset[L]{2.25in}

\usepackage{color}

\definecolor{Gray}{gray}{.25}

\usepackage{graphicx}

\usepackage{sidecap}

\usepackage{wrapfig}
\usepackage[pscoord]{eso-pic}
\usepackage[fulladjust]{marginnote}
\reversemarginpar

\usepackage{amsthm}
\usepackage{algorithm, algorithmic}
\usepackage{graphicx}

\newtheorem{thm}{Theorem}[section]
\theoremstyle{definition}
\newtheorem{dfn}{Definition}[section]
\theoremstyle{remark}
\theoremstyle{plain}
\newtheorem{lem}[thm]{Lemma}
\newtheorem{col}[thm]{Corollary}

\begin{document}

\begin{flushleft}
{\Large
\textbf\newline{Deterministic Logarithmic Completeness in the Distributed Sleeping Model \footnote{This work was supported by the Open University of Israel research fund.} }
}
\newline
\\
Leonid Barenboim \textsuperscript{$a$},
Tzalik Maimon \textsuperscript{$b$}
\\
\bigskip
{\bf $^a$ The Open University of Israel.} Email: leonidb@openu.ac.il
\\
{\bf $^b$ Ben-Gurion University of The Negev, Israel.} Email: tzalik@post.bgu.ac.il
\\
\end{flushleft}

\begin{abstract}
In this paper we provide a deterministic scheme for solving any decidable problem in the distributed {\em sleeping model}. The sleeping model \cite{KPSY11,CGP20} is a generalization of the standard message-passing model, with an additional capability of network nodes to enter a sleeping state occasionally. As long as a vertex is in the awake state, it is similar to the standard message-passing setting.
However, when a vertex is asleep it cannot receive or send messages in the network nor can it perform internal computations. On the other hand, sleeping rounds do not count towards {\em awake complexity.} Awake complexity is the main complexity measurement in this setting, which is the number of awake rounds a vertex spends during an execution. In this paper we devise algorithms with worst-case guarantees on the awake complexity. 

We devise a deterministic scheme with awake complexity of $O(\log n)$ for solving any decidable problem in this model by constructing a structure we call {\em Distributed Layered Tree}. This structure turns out to be very powerful in the sleeping model, since it allows one to collect the entire graph information within a constant number of awake rounds. Moreover, we  prove that our general technique cannot be improved in this model, by showing that the construction of distributed layered trees itself requires $\Omega(\log n)$ awake rounds. This is obtained by a reduction from message-complexity lower bounds, which is of independent interest.
Furthermore, our scheme also works in the $\mathcal{CONGEST}$ setting where we are limited to messages of size at most $O(\log n)$ bits. This result is shown for a certain class of problems, which contains problems of great interest in the research of the distributed setting. Examples for problems we can solve under this limitation are leader election, computing exact number of edges and average degree. 

Another result we obtain in this work is a deterministic scheme for solving any problem from a class of problems, denoted {\bf O-LOCAL}, in $O(\log \Delta + \log^*n)$ awake rounds. This class contains various well-studied problems, such as MIS and $(\Delta+1)$-vertex-coloring. Our main structure in this case is a tree as well, but is sharply different from a distributed layered tree. In particular, it is constructed in the local memory of each processor, rather than distributively. Nevertheless, it provides an efficient synchronization scheme for problems of the {\bf O-LOCAL} class.
\end{abstract}


\section{Introduction}

What can be computed within logarithmic complexity has been one of the most fundamental questions in distributed and parallel computing, since the initiation of the study of parallel algorithms in the eighties. Various problems were shown to belong to the NC class back then, i.e., the class of problems that can be solved in polylogarithmic time by a polynomial number of machines. This includes several fundamental problems, namely, $(\Delta + 1)$-coloring, Maximal Independent Set and Maximal Matching. All of these problems admit deterministic logarithmic parallel algorithms. In the distributed setting, however, these problems turned out to be much more challenging, if one aims at a deterministic solution. The first deterministic polylogarithmic solution was an $O(\log^7 n)$ time algorithm for the problem of Maximal Matching, obtained by Hanckowiak, Karonski and Panconesi \cite{HKP98}. More recently, polylogarithmic deterministic $(\Delta^{1 + \epsilon})$-coloring was obtained by
Barenboim and Elkin \cite{BE11}. The problem of $(2\Delta - 1)$-edge-coloring was provided with an $O(\log^7 \Delta \log n)$-rounds deterministic algorithm by Fischer, Ghaffari and Kuhn \cite{FGK17}. Recently, a plethora of results were published in this field, with various improvements to the aforementioned algorithms. See, e.g., \cite{BEM17, BKO20, Fi20, K20}, and references therein.
In a very recent breakthrough, a wide class of problems have been solved using deterministic polylogarithmic number of rounds \cite{RG20}, including $(\Delta + 1)$-coloring and Maximal Independent Set. This was achieved by providing an efficient algorithm for the $(O(\log n), O(\log n))$-Network-Decomposition problem, which is complete in this class.

\subsection{Our Results}

In the current paper we investigate yet another distributed setting,  namely, the {\em Sleeping Setting}. Several variants of this setting have attracted the attention of researchers recently \cite{BT19, CGP20, DHW98, F20, GKKPS08, KPSY11}. The particular setting and complexity measure we consider in this paper were introduced by Chatterjee, Gmyr and Pandurangan \cite{CGP20} in PODC'20. This sleeping setting is similar to the standard distributed $\mathcal{LOCAL}$ setting \cite{L87}, but has an additional capability, as follows. In the sleeping setting, the vertices of the network graph can decide in each round to be in one of two states; a "sleep" state or an "awake" state. If all the vertices are awake all the time, the setting is identical to the standard $\mathcal{LOCAL}$ setting. However, the capability of entering a "sleep" state is where a vertex cannot receive or send messages in the network, neither can it perform internal computations. Consequently, such rounds do not consume the resources of that vertex, and shall not be counted towards a complexity measurement that aims at optimizing resource consumption. Indeed, in this setting a main complexity measurement takes into account only awake rounds. Specifically, the worst-case awake complexity of an algorithm in the sleeping setting is the worst-case number of rounds in which any single vertex is awake.
In PODC'20, Chatterjee, Gmyr and Pandurangan \cite{CGP20} presented a Maximal Independent Set randomized algorithm with expected awake complexity of $O(1)$. Its high-probability awake complexity is $O(\log n)$, and its worst-case awake complexity is polylogarithmic. 

This work raised the following two important questions: \\ 
(1) Can MIS be solved within deterministic logarithmic awake complexity? \\
(2) Can additional problems be solved within such complexity? \\

\noindent In the current paper we answer these questions in the affirmative. But much more generally, we show that {\em any decidable problem} can be solved within deterministic logarithmic awake complexity in the distributed sleeping setting. Namely, a decidable problem is any computational problem that has a sequential deterministic algorithm that provides a correct solution within a finite sequential running time (as large as one wishes).

Note that undecidable problems in the sequential setting are also undecidable in the different variants of distributed settings.

For the purpose of solving decidable problems, we present a new structure, namely, a {\em Distributed Layered Tree} (DLT). We show that if one is able to compute a distributed layered tree, then any decidable problem can be solved within additional awake complexity of $O(1)$. This is because a DLT allows each vertex to obtain all the information of the input graph in a constant number of awake rounds, and then any decidable problem can be solved locally and consistently by all vertices using a sequential algorithm. We also prove that DLT itself can be solved in $O(\log n)$ awake rounds. In particular, this provides a deterministic logarithmic solution to the fundamental Broadcast problem. This improves the best previously-known awake complexity of this problem in the sleeping setting, due to Chang et al. \cite{CDHHLP18}, by at least a quadratic factor. We note that the broadcast algorithm of Chang et al. was devised for settings with additional requirements, i.e., it is more general than Broadcast in the sleeping setting. Nevertheless, it was still the state-of-the-art even in the sleeping model. Our improvement applies specifically to the sleeping model. 

A natural question is how difficult the construction of DLT is.
We answer this
by proving a lower bound of $\Omega(\log n)$ awake rounds for the DLT problem. This lower bound is obtained by a simple but powerful tool of a reduction from {\em message complexity} lower bounds in the $\mathcal{LOCAL}$ model. With this lower bound, given that the DLT problem itself is a decidable problem, we obtain a tight deterministic bound of $\Theta(\log n)$ worst-case awake time on the class of decidable DLT-hard problems \footnote[1]{A {\em DLT-hard} problem is a problem whose solution provides a DLT within additional $o(\log n)$ awake rounds.} in the sleeping model.

An additional direction of ours is the analysis of a class we define as the {\bf O-LOCAL} class. This is a class of problems that can be solved using an acyclic orientation of the edges, by choosing a solution for each vertex after all vertices reachable from it have computed their solution, and as a function of these solutions. A notable example is $(\Delta + 1)$-coloring, where each color can be selected once all neighbors on outgoing edges have selected their own colors, such that the color does not conflict with any of them. Another example is MIS, where each decision is made after all outgoing neighbors have made their decisions. We show that any problem that belongs to this class has deterministic worst-case awake complexity of $O(\log \Delta + \log^* n)$. 

In addition to the number of awake rounds, which is the main complexity measurement in this setting, we are also interested in optimizing the overall number of communication rounds. Since the DLT can be used to solve any decidable problem, it follows that certain such problems require $\Omega(n)$ communication rounds. (These are the global problems of the ordinary distributed setting. For example, the leader election is such a problem.) We investigate how close we can get to this lower bound with an algorithm of $\tilde{O}(n)$ rounds for the distributed layered tree problem. While our basic algorithm requires $O(n^2 \log n)$ communication rounds, a more sophisticated version requires only $O(n \log n \log^*n)$ communication rounds. This comes at a price of increasing the worst-case awake complexity, but only by a factor of $O(\log^* n)$.

\subsection{The Sleeping Setting} 

The sleeping setting represents the need for energy-efficient algorithms in ad hoc, wireless and sensor networks. In such networks, the energy consumption depends on the amount of time a vertex is actively communicating or performing calculations. More importantly, significant energy is spent even if a node is idle, but awake. It was shown in experiments that the energy consumption in an idle state is only slightly smaller than when a node is active \cite{FN01, ZK05}. This is in contrast to the sleeping state, in which energy consumption is decreased drastically. Thus, if a node may enter a sleeping mode to save energy during the course of an algorithm, one can significantly improve the energy consumption of the network during the execution of an algorithm. 

The sleeping model is a formulation of this premise, and is a generalization of the traditional $\mathcal{LOCAL}$ model. In the sleeping model, similarly to the $\mathcal{LOCAL}$ model, a communication network is represented by  an $n$-vertex graph $G = (V,E)$, where vertices represent processors, and edges represent communication links. There is a global clock, and computation proceeds in synchronous discrete rounds. In each round a vertex can be in either of the two states, "sleep" or "awake". (While in the $\mathcal{LOCAL}$ model the vertices are only in the "awake" state.) If a vertex is in the "awake" state in a certain round, it can perform local computations as well as sending and receiving messages in that round. On the other hand, in a round of a "sleep" state, a vertex cannot send or receive messages, and messages sent to it by other vertices are lost. It also cannot perform any internal computations. A vertex decides a-priori about entering a "sleeping" state. That is, in order to enter a sleeping state in a certain round $k$, either the vertex decides about it in an awake round $k' < k$, or such a decision is hard-coded in the algorithm, and is known before its execution. Nodes in the "sleep" state consume almost no energy, and thus shall not be counted towards the energy efficiency analysis. 

Initially, vertices know the number of vertices $n$, or an upper bound $\hat{n} \geq n$. The IDs of vertices are unique and belong to the set $[\hat{n}]$. Even if $\hat{n} >> n$, the awake complexity of our algorithms is not affected at all. The overall number of clock rounds, however, may be affected. In this case $n$ should be replaced by $\hat{n}$ in the clock complexity bounds. In some of our algorithms, however, the dependency on $n$ and $\hat{n}$ can be made as mild as the log-star function. See Section 4. 
\\

\noindent {\bf The main efficiency measures in the Sleeping Model}. The measurements for the performance of an algorithm in the sleeping model were first mentioned by Chatterjee, Gmyr and Pandurangan in \cite{CGP20}. For a distributed algorithm with input graph $G = (V, E)$ in the sleeping model, two types of complexity measurements are defined. One is {\em node-averaged awake complexity} in which, for a node $v \in V(G)$, define $a(v)$ as the number of rounds $v$ spends in the "awake" state until the end of the algorithm. The node-average awake complexity is the average $\frac{1}{n} \sum_{v \in V(G)} a(v)$. \\
The second efficiency measurement is the {\em worst-case awake complexity}. This is defined as $\max_{v \in V(G)} a(v)$ and is a stronger requirement than the node-averaged awake complexity. In this paper we focus entirely on the worst-case efficiency measurement.

\subsection{Our Techniques}

\subsubsection{Upper Bound}

Our main technical tool is the construction of a {\em Distributed Layered Tree}. (We denote it shortly by {\em DLT}.) A DLT is a rooted tree where the vertices are labeled, such that each vertex has a greater label than that of its parent, according to a given order. Moreover, in a DLT each vertex knows its own label and the label of its parent. This knowledge of the label of the parent is not trivial in the sleeping model since passing this information between parent and child requires both of them to be in an awake state. Therefore, this knowledge and hierarchy of IDs throughout the tree makes DLTs are very powerful structures in the sleeping setting. Indeed, once such a tree is computed, the information of the entire graph can be learned by all vertices within $O(1)$ awake complexity, as follows. For a non-root vertex $v \in V(G)$, let $L(v)$ and $L(p(v))$ be the labels of $v$ and the parent of $v$ in the DLT, respectively. Each non-root vertex $v \in V(G)$ is awake only in rounds $L(p(v))$ and $L(v)$. The root $r$ awakes only in round $L(r)$. This way the root is able to perform a broadcast to all the vertices of the tree. Each vertex $v \in V(G)$ receives the information from the root in round $L(p(v))$ (this information has arrived to the parent of $v$ in an earlier stage) and passes it to its children in round $L(v)$. Indeed, in this round $v$ and all its children are awake. In a similar way, a convergecast in the DLT can be performed. We choose some label $n'$ which is greater than all vertex labels in the DLT. Each vertex $v \in V(G)$ awakens in rounds $n' - L(p(v))$ and $n' - L(v)$. This way, information from the leaves propagates towards the root. In round $n' - L(v)$ a vertex $v$ receives information from all its children, and in a later stage, in round $n' - L(p(v))$, the vertex forwards the information to its parent. Note that indeed $n' - L(v) < n' - L(p(v))$, since $L(v) > L(p(v))$, according to the definition of a DLT. A formal proof for the running time of the broadcast and convergecast procedures in a DLT can be found in Lemma \ref{lem:DLTprop} in Section 2. 

Thus, it becomes possible to perform broadcast and convergecast in the tree within $4$ awake rounds per vertex. A broadcast and convergecast in a tree allows each vertex to obtain the entire information stored in the tree. Since the tree spans the input graph, the entire information of the graph is obtained. Then, any decidable problem can be solved using the same deterministic algorithm internally in all vertices of the graph. Finally, each vertex deduces its part in the solution. This execution that is performed internally, does not require any additional rounds of distributed communication and is considered as a single awake round in terms of the sleeping model. To summarize this discussion, a DLT makes it possible to solve any decidable problem within 5 awake rounds per vertex. 

The ability to solve any decidable problem within a constant awake complexity suggests that the computation of a DLT is an ultimate goal in the sleeping setting. Thus establishing the efficiency of this construction is of great interest. We construct a DLT within $O(\log n)$ awake rounds as follows. We begin with $n$ singleton DLTs, where each vertex of the input graph is a DLT in a trivial way. Then, we perform $O(\log n)$ connection phases in which the trees are merged. Each phase requires at most $O(1)$ awake rounds from each vertex. The number of DLT trees in each phase is at least halved.
After $O(\log n)$ phases, a single tree remains. This DLT contains all the vertices of the input graph. Thus, it is the DLT of the entire input graph $G$. 

The high-level idea of our algorithm is somewhat similar to the celebrated algorithm of Gallager, Humblet and Spira for minimum spanning trees \cite{GHS83}. But the construction is fundamentally different. While GHS finds an {\em existing} subgraph that is an MST, our technique gradually builds trees that contain {\em new data}. These are new ID assignments that make it possible to have progress with trees formation. In each iteration trees are merged and IDs are reassigned, until a single DLT of the entire input graph is achieved. This tree has the desired IDs, according to the definition of the DLT, as a result of the ID recomputation made in each iteration.

\subsubsection{Lower Bound}

Once we establish an upper bound on the awake complexity for constructing DLT, we turn to examining lower bounds. We note that an ordinary lower bound technique may not work for the sleeping setting. This is because the standard techniques in the distributed setting deal with what information can be obtained within a certain number of rounds. That is, within $r$ rounds, for an integer $r$, each vertex can learn its $r$-hop neighborhood. Then arguments of indistinguishably of views are used. (That is, vertices that must make distinct decisions are unable to do that, if their $r$-hop-neighborhoods are identical. In this case, $r$ rounds are not sufficient to solve a certain problem.) However, such arguments do not work in the sleeping setting. Indeed, within $O(1)$ awake rounds the entire graph can be learned on certain occasions. Thus, algorithms with $r$ awake rounds are not limited to obtaining knowledge of $r$-hop-neighborhoods. 

As a consequence of the latter phenomenon, we investigate alternative ways to prove lower bounds. We introduce a quite powerful technique that allows one to transfer lower bounds on message complexity into lower bounds for rounds in the sleeping setting. Indeed, if $t n$ messages must be sent in a ring network to solve a certain problem, for an integer $t > 1$, then $t/2$ awake rounds are required for any algorithm that solves the problem in the sleeping setting. Otherwise, if $t' < t/2$ awake rounds are possible, all the messages in each round can be concatenated, and thus each vertex sends up to $t'$ messages to each of its two neighbors in the ring, during the $t'$ rounds it is awake. The number of messages per vertex becomes $2t' < t$, and the overall number of messages passed is thus smaller than $t n$, which contradicts the assumption that at least $t n$ messages must be sent. A formal proof for this claim can be found in Lemma \ref{lem:minMSGs} in Section 3.

We employ this idea with the known lower bound of $\Omega(n \log n)$ for message complexity of leader election in rings \cite{FL87}. Since a DLT allows, in particular, to solve leader election, we deduce that $\Omega(\log n)$ awake rounds are required. Otherwise, it would be possible to solve leader election in rings within fewer than $\Theta(n \log n)$ messages. We note that the lower bound of \cite{FL87} holds for a given number of rounds assuming the IDs are sufficiently large. Our upper bounds on awake complexity, on the other hand, do not rely on the range of IDs, but only on the number of vertices in the graph. No matter how large the IDs are, on an $n$-vertex graph the awake complexity for constructing a DLT is $O(\log n)$. The overall number of clock rounds (awake and asleep) do depend on the range of identifiers, but the dependency can be made as low as the log-star function, by using the coloring algorithm of Linial \cite{L87}. Our algorithm is applicable also with proper coloring of components, not necessarily with distinct IDs. Consequently, for any given number of clock rounds (awake and asleep), which upper bounds the ordinary running time of an algorithm, there exists a sufficiently large range of IDs, such that the awake complexity $O(\log n)$ of our algorithm is tight.

\subsubsection{Improved Upper-Bound for O-LOCAL problems}  \label{subsec:introOLOCAL}

O-LOCAL problems are those that can be solved sequentially, according to an acyclic orientation provided with the input graph, such that each vertex decision is made after all vertices emanating from it have made their own decisions, and as a function of these decisions. (Note that directing edges from endpoints of smaller IDs to larger IDs provides such an orientation.) For these kind of problems, we employ a technique that is quite different from that of a DLT, and obtain awake complexity of $O(\log \Delta + \log^* n)$. We still employ a tree construction, but this time it is more sophisticated than a DLT. On the other hand, it is constructed internally by each vertex, and is the same in all vertices. The algorithm starts by a distributed computation of an $O(\Delta^2)$-coloring of the input graph within $O(\log^* n)$ time. Then, each vertex constructs internally a binary search tree, whose leaves are the possible $O(\Delta^2)$-colors. (The colors are not consecutive, and inner nodes have integer values between the values of the colors.) Next, each vertex decides to wake-up in the rounds whose numbers appear in the path from the leaf of their color and the root. We prove that one of these rounds occurs after all neighbors of smaller colors have made their decisions. Moreover, by that round these vertices have communicated their decision to the vertex. Consequently, it may compute its own decision. Since the depth of the tree is $O(\log \Delta)$, this requires only $O(\log \Delta)$ awake-rounds per vertex.

\subsection{Related Work}

The distributed $\mathcal{LOCAL}$ model was formalized by Linial in his seminal paper \cite{L87} from 1987. This paper also provides a deterministic $O(\Delta^2)$-coloring algorithm with $O(\log^* n)$ round-complexity, as well as a matching lower bound. Since then, a plethora of distributed graph algorithms has been obtained in numerous works. See, e.g., the survey of \cite{S13} and references therein.

The sleeping setting has been intensively studied in the area of Computer Networks \cite{DHHV05, PXW09, TGAAA17, WYHD19}. In Distributed Computing the problem of broadcast in sleeping radio networks was studied by King, Phillips, Saya and Young \cite{KPSY11}. The problem of clock synchronization in networks with sleeping processors and variable initial wake-up times was studied by Bradonjic, Kohler and Ostrovsky \cite{BKO12}, and by Barenboim, Dolev and Ostrovsky \cite{BDO14}. A special type of the sleeping model, in which processors are initially awake, and eventually enter a permanent sleeping state, was formalized by Feuilloley \cite{F17,F20}. An important efficiency measurement in this setting is the {\em vertex-averaged} awake complexity. This setting was further studied by Barenboim and Tzur \cite{BT19}, who obtained various symmetry-breaking algorithms with improved vertex-averaged complexity.

The awake complexity of various problems has been also studied in radio models of general graphs (rather than unit disk graphs). In particular, several important results were achieved by Chang et al. in PODC'18 \cite{CDHHLP18}. That work considered Broadcast and related problems in several radio models, that can be seen as the sleeping model with additional restrictions. Specifically, in the model that is the closest to the sleeping model, the vertices are able to either transmit or listen in an awake round, but not both. (In other words, this is a half-duplex communication model, while the sleeping model is full-duplex.) There are also even more restricted models studied in \cite{CDHHLP18}, in which vertices cannot receive messages from  multiple neighbors in parallel.

Since the results of Broadcast \cite{CDHHLP18} are applicable to the sleeping model, and they are the state-of-the-art even in this model that we consider in the current paper, a comparison between them and our results is in place. The Broadcast algorithm of \cite{CDHHLP18} with the best deterministic awake complexity has efficiency $O(\log n \cdot \log N)$, where $N \geq n$ is the largest identifier. Our results, on the other hand, provide a deterministic Broadcast algorithm in the sleeping setting with awake complexity of $O(\log n)$. This is at least a quadratic improvement in the sleeping setting. We stress that our awake complexity is not affected by the size of identifiers, and remains $O(\log n)$, no matter how large $N$ is. The Broadcast algorithm of \cite{CDHHLP18} constructs trees that partition the graph into layers, but these trees are very different from our DLTs, both in their structure and in the techniques for achieving them. Specifically, in \cite{CDHHLP18} each vertex in layer $i = 1,2,...$ in the tree has a neighbor of layer $i - 1$. On the other hand, a DLT does not necessarily have this property. (This is because layer $i$ in a DLT consists of all vertices labeled by $i$ in the tree, which are not necessarily at distance $i$ from the root.) In addition, the tree construction in \cite{CDHHLP18} is based on ruling sets, while our techniques are considerably different.

The class {\bf O-LOCAL} of problems that we mentioned in Section \ref{subsec:introOLOCAL} is inspired by the class {\bf P-SLOCAL} which was first defined by Ghaffari, Kuhn and Maus \cite{GKM17}. This class consists of all problems that can be solved as follows. Given an acyclic orientation, the output of each vertex is determined sequentially, according to the orientation, after vertices on outgoing edges have made their decisions. The decision of each vertex is based on information from a polylogarithmic radius around it. The {\bf O-LOCAL} class is similar to the {\bf P-SLOCAL} class, except that instead of examining a polylogarithmic- radius neighborhood around a vertex, its neighbors on outgoing edges are examined. (And, more generally, all vertices on consistently oriented paths emanating from a vertex are inspected.)

\section{Distributed Layered Trees} \label{sec:DLTmain}

In this section we describe our method with which one can solve any decidable distributed problem in the sleeping model. We describe the construction of a certain kind of a spanning tree, called a {\em Distributed Layered Tree}, defined as follows. Each vertex $v$ in the tree is labeled with a label $L(v)$. These labels must have some predefined order, such that they can be mapped to natural numbers. 
The labeling is such that the label of each vertex, besides the root, is larger than the label of its parent, and each vertex knows the label of its parent. These two requirements of the spanning tree allow us to perform broadcast across the spanning tree in a fashion where each vertex is awake for exactly 2 rounds. This is also true for a convergecast procedure. Consequently, given such a tree, the root can learn the entire input graph, compute a solution for any decidable problem internally, and broadcast it to all vertices, all within a constant number of awake rounds. This is done in the following way. For a broadcast procedure, we start with a message from the root of the tree through the tree, where each vertex $v$ is awake for just 2 rounds. Namely, $v$ awakes in round $L(v)$ and round $L(p(v))$, where $p(v)$ is the parent of $v$ in the spanning tree. In other rounds $v$ is asleep. This ensures that a message sent from the root will propagate in the tree and eventually arrive to all vertices of the graph while having each vertex awaken exactly twice. In the same manner we perform the convergecast procedure, where each child $v$ sends its message to its parent at the round $n' - L(p(v))$ for some $n' \geq n$. Again we have each vertex $v$ active for only two rounds, specifically, $n' - L(v)$ and $n' - L(p(v))$. 

Using this method one can have the root of the spanning tree compute the solution for any decidable problem $\mathcal{P}$ deterministically and broadcast this solution back through the tree to all the vertices in the input graph $G$ with $O(1)$ worst-case complexity in the sleeping model. Therefore, our main goal is obtaining a distributed sleeping algorithm for computing such a layered tree. We begin with a formal definition of notations and the definition of a Distributed Layered Tree. Note that we refer here to lexicographic order. For sake of definition we do not limit or define this lexicographic order since it does not matter for the definition of a DLT. One can build a DLT with any order that can be mapped to natural numbers. We define a lexicographic order that serves our purposes in Section 2.1.

\noindent {\bf Vertex Label $L(v)$.} The label of a vertex $v$ is denoted by $L(v)$. The labels are taken from a range of labels which has a lexicographic order. \\
\noindent {\bf Tree Label $L(\mathcal{T})$.} The label of a tree $\mathcal{T}$ is denoted by $L(\mathcal{T})$. The labels are taken from a range of labels which has a lexicographic order. The label of a tree is defined to be the label of the root of the tree. Hence, that label can always be found in the memory of the root.

\begin{dfn}
{\bf A Distributed Layered Tree (DLT).} A DLT is a rooted oriented labeled spanning tree with two properties, with respect to some lexicographical order:
\begin{enumerate}
    \item For each vertex $v$ and a parent $p(v)$, the label of $v$ is greater than the label of $p(v)$ in the lexicographical order.
    \item $v$ has knowledge of the label of $p(v)$.
\end{enumerate}
\end{dfn}

\noindent As we show in the next lemma, DLTs are useful for distributing data across the graph in an efficient way.

\begin{lem} \label{lem:DLTprop}
Given a DLT, the procedures of broadcast and convergecast across the entire tree take exactly 2 rounds each in the sleeping model.
\end{lem}
\begin{proof} \textbf{Broadcast.} Each vertex $v$ is awake in rounds $L(v)$ and $L(p(v))$. As a vertex $v$ awakes in round $L(v)$ and broadcasts a message, each of its children $u$ in the tree are awake in round $L(p(u)) = L(v)$ and thus can receive the message from its parent. Therefore, a message sent by the root propagates throughout the tree until it reaches all leaves. \\
\textbf{Convergecast.} Let $n'$ be an integer, such that $n' > L(v)$, for all $v \in V$, and $n'$ is known to all vertices. Each vertex $v$ is awake in rounds $n' - L(v)$ and $n' - L(p(v))$. If a child $v$ has a message to pass to its parent $p(v)$, it awaits round $n' - L(p(v))$ and then $v$ sends the message to $p(v)$. In that round $p(v)$ is awake and ready to receive the messages from all its children. Each vertex $v$ already has knowledge of the subtree rooted at it, since for each vertex $u$ in the subtree in which $v$ is the root, $L(u) > L(v) > L(p(v))$ and thus the round $n' - L(u)$ comes before the round $n' - L(p(v))$. Thus the message propagates up the tree all the way to the root.
\end{proof}

\noindent We note that in the proof above we require for an integer $n'$ to be larger than all labels of all vertices in the tree. Since the labels are required to be taken from a range with lexicographic order, and we have knowledge of the size of this range, $n'$ can be chosen appropriately. 
We describe the lexicographic order in Section \ref{sec:conPhase}, as well as describing how each vertex has knowledge of ranges from which labels are selected.

\subsection{The Connection Phases} \label{sec:conPhase}

Our algorithm for constructing a DLT starts with a graph where each vertex is considered to be a singleton tree. The initial label of each such tree $\mathcal{T}_v$, \ $v \in V$, is determined by the ID of its vertex $v$ as follows. $L(\mathcal{T}_v) = ID(v)$. The vertex label is set as $L(v) = \langle ID(v) , 0 \rangle$. (Two coordinates are used, since trees are going to be merged and have many vertices. Then the left-hand coordinate is going to be the same for all vertices in a tree, while the right-hand coordinates may differ. Also, distinct trees will have distinct left-hand coordinate.) Each of these $n$ singleton trees is a DLT in a trivial way. Our goal is merging these trees in stages, so that eventually a single DLT remains that spans the entire input graph.
Assume we have a forest of DLTs. Initially, we have a forest of $n$ single-vertex trees. During the connection phases, we enforce a rule regarding the representation of the labels of the vertices.
Let $\mathcal{T}$ be a tree in our forest. The DLT label of $\mathcal{T}$, $L(\mathcal{T})$ is an integer number. The label of each vertex $v \in \mathcal{T}$ is set as $\langle L(\mathcal{T}), l \rangle$, where $l$ is a number assigned to $v$, as described in Section \ref{subsec:stage1}. In what follows we describe the ordering of the vertex labels. The left coordinate is considered to be the more significant one among the two. That is, $\langle a, b \rangle < \langle c,d \rangle$ iff $a < c$ or ($a = c$ and $b < d$). Note that the requirements on labels hold. That is, the labels have an ordering and the root of the tree can deduce the tree label from its own label by referring to the first coordinate. Next, we describe how our algorithm produces connections between the trees. We do this in two stages.

\subsubsection{Connection Stage One - Several DLTs into a single DLT}  \label{subsec:stage1}

In this stage our goal, for each tree $\mathcal{T}$, is finding an edge $(u,w)$ that connects $\mathcal{T}$ to a neighbor DLT $\hat{\mathcal{T}}$, such that $L(\mathcal{T}) > L(\hat{\mathcal{T}})$. In this sense, $u \in \mathcal{T}$ and $w \in \hat{\mathcal{T}}$. Using this edge we connect $\mathcal{T}$ with $\hat{\mathcal{T}}$, such that $w$ becomes the parent of $u$. In the case that $\mathcal{T}$ is a single vertex $v$, we simply choose the neighbor of $v$ in $G$ with a label smaller than that of $v$. We note that there may be a case that no such edge is found, since $\mathcal{T}$ is a DLT with local minimum label. We handle such a case in Section \ref{subsec:conStage2}. If $\mathcal{T}$ contains more than one vertex, we perform a convergecast from all the neighbors of all vertices in $\mathcal{T}$ to the root of $\mathcal{T}$. Consequently, the root of $\mathcal{T}$ learns the structure of $\mathcal{T}$ and the set of edges that connect $\mathcal{T}$ with other trees. The root chooses the edge $(u, w)$ which connects $\mathcal{T}$ to a vertex $w \in \hat{\mathcal{T}}$, where $\hat{\mathcal{T}} \neq \mathcal{T}$ is a neighboring DLT of $\mathcal{T}$. As explained above, the choice is made such that $L(\hat{\mathcal{T}}) < L(\mathcal{T})$. Recall that at this point $L(\hat{\mathcal{T}})$ appears in the first coordinate of all the labels of the vertices in $\hat{\mathcal{T}}$, and hence the root of $\mathcal{T}$ has knowledge of all the labels of its neighboring DLTs. 

Internally, the root of $\mathcal{T}$ calculates a new label arrangement of the vertices of $\mathcal{T}$, such that the vertex $u$ that was chosen above (that is connected to $w \in \hat{\mathcal{T}}$) becomes the new root, and $\mathcal{T}$ remains a DLT, under the new label arrangement. This requires a new orientation of $\mathcal{T}$. This new oriented tree is denoted $\mathcal{T}'$. The label arrangement of $\mathcal{T}'$ is obtained in the following way. Note that the DLT label of $\mathcal{T}$ is $L(\mathcal{T})$. The new root $u$ sets its label to be $\langle L(\mathcal{T}), 0 \rangle$. The rest of the vertices in $\mathcal{T}'$ are assigned labels in the following way. Each vertex $v$ is assigned the label $\langle L(\mathcal{T}), l \rangle$ where $l$ is the distance of $v$ from $u$ in $\mathcal{T}'$. Specifically, $l$ is the level of $v$ in $\mathcal{T}'$. It follows that each level of $\mathcal{T}'$ has labels smaller than the labels in the next level of the tree. Once this internal computation is done, $r_{\mathcal{T}}$ sends the resulting label arrangement in a broadcast procedure to all vertices in $\mathcal{T}'$. Note that $u$ becomes the new root, instead of $r_{\mathcal{T}}$, from now until further notice. 

But now we are posed with the problem that neighbors of the vertices of the tree $\mathcal{T}'$ do not know the labels of their neighbors, which is required for the next connection phase as well as the second requirement for a DLT. This is solved by awakening the entire graph for one round. Each vertex sends its knowledge to all of its neighbors while receiving the knowledge from all of its neighbors in the same round. Then all vertices of the graph switch to asleep status. The round in which this happens is determined by all vertices before the beginning of the execution. Once all vertices determine the starting time of such a phase, they all wake up after $c n'$ rounds, for a constant $c$, such that $c n'$ bounds from above the execution time of this phase. In the following lemma we prove that this connection procedure connects several DLTs into one DLT.

\begin{lem} \label{lem:isDLT}
Let $C$ be a connected component in our graph produced at this stage. Then $C$ is a DLT.
\end{lem}
\begin{proof}
First, $C$ is a tree since each edge in $C$ is oriented from a vertex with a high label to a vertex with a lower label. Thus consistently oriented cycles are not possible. Moreover, each vertex has a single parent, the one which it chose to connect to. Thus cycles are not possible at all. (A cycle is either consistently oriented or with a vertex that selects two neighbors in the cycle.) Furthermore, there is a tree $T_m$ with the minimal DLT label among all DLTs that compose $C$. This root did not choose any edge to connect through since the root label is a local minimum. This local minimum label appears in the first coordinate of all the labels of its vertices making them local minima in their respective neighborhoods. Thus, the root of $T_m$ is the root of the tree $C$.

Secondly, we made sure that each child has a label greater than the label of its parent by the way we chose the connections. Let $v,u$ be two neighbors in $C$ such that $u$ is the parent of $v$. If $u$ and $v$ originated from the same tree $T$ at the beginning of the connection phase then the root of $T$ assigned them new labels in the connection phase such that $L(u) < L(v)$. Otherwise, $u \in \hat{\mathcal{T}}, v \in \mathcal{T}$, and $v$ becomes the root of its original tree and sets its own label as $\langle L(\mathcal{T}), 0 \rangle$. Since we chose $\hat{\mathcal{T}}$ such that $L(\hat{\mathcal{T}}) < L(\mathcal{T})$ we have that $L(u) < L(v)$ in this case as well. 

Thirdly, by awakening all the vertices in the graph for one round and sending labels to 1-hop neighborhoods, we made sure that vertices has knowledge of labels in their 1-hop neighborhood (including their respective parents in $C$).\\
Thus, $C$ is a tree, where each child has a label greater than the label of its parent and each child knows the label of its parent. Thus, by definition, $C$ is a DLT.
\end{proof}

We would like to note here that at the end of each connection phase each DLT $C''$ in the new forest has a distinct DLT label. This is due to the fact that the DLT label of $C''$ is set from one of the labels of the trees composing $C''$. If each DLT had a distinct label at the start of the connection phase, it is clear we preserved this property also once the connection phase is done. Thus, it is also true in the start of the next connection phase. Note that we start the algorithm with each DLT having a distinct DLT label since we set those labels according to the IDs of the vertices of $G$.

\subsubsection{Connection Stage Two - Connecting Local Minimums} \label{subsec:conStage2}

We now make a second connection step as part of the connection phase. The motivation in this step is that there might be trees where the DLT labels are local minima and thus those trees did not connect to any other tree. They might also not have been chosen by any other tree. If such is the case, these trees made it through stage one of the connection phase without connecting to any other component. 

We would like to at least halve the number of trees in each phase and have at most $O(\log n)$ connection phases, we connect these local-minimum label trees to other components. Let $T_m$ be such a local-minimum label tree. First, we would like to verify that $T_m$ indeed has no connections to other components from the previous stage. If some other component has selected $T_m$ earlier, we no longer need to handle it, even though $T_m$ is a minimum label tree. \\
We only need to handle trees with no connection to other trees. We check this by performing a convergecast in each connected component sending signals from the leaves to the root. If $r_{T_m}$ received a signal from vertices that previously did not belong to $T_m$ then $T_m$ is not a problematic tree and was chosen by another tree in the first stage. This signal can be, for example, the label of a leaf. Since the label contains the root of the tree to which the leaf belongs (before the trees are merged), then $r_{T_m}$ can deduce that the leaf was not part of $T_m$ before the first stage of the connection phase. In this case we do nothing with $T_m$. 

The other option is where $T_m$ was not connected to any other component. In this case, we add such a connection. To this end, $r_{T_m}$ chooses an edge $(v, u)$ to connect to a DLT $C$ arbitrarily. Since it is the only edge connecting it to other DLTs, the result is a tree. A simple move now would be to make $C$ a subtree of the DLT $T_m$. But we note that $T_m$ might not be the only tree that arbitrarily chose to connect to $C$. Since there can be more than one local minimum DLT that chooses $C$, we need to make $C$ the parent of all those DLTs which chose to connect to it. Otherwise, $C$ will become the child-DLT of more than one DLT which breaks the structure of a directed tree we aspire for. \\
Therefore, in order for the result to be a DLT we need to make $T_m$ the sub-DLT of $C$ instead. That is, given that $(v, u)$ is the edge connecting $T_m$ to $C$, we turn $v$ into the root of $T_m$ and make $v$ the child of $u$ in $C \cup T_m$. Doing so poses a problem since $L(v) < L(u)$ and this violates the requirements for a DLT in the resulting component. We solve this by waking up the entire graph for a single round and have $v$ and $u$ exchange information. 

After this round, the information about the local-minimum label trees that asked to join $C$ is located in vertices of the component $C$. This information is then delivered to the root of $C$, $r_C$ by a convergecast procedure. $r_C$ performs label reassignment in the same way as in Section \ref{subsec:stage1}. Specifically, a single BFS is computed for all vertices in $C$ and the local minimum label trees $T_m$ that connected to $C$. Then reassignment of these labels is made according to the levels of the BFS tree. This results in labels of the from $\langle ID(r_C), l \rangle$, such that the first coordinate is the same for all vertices in the connected component and $l$ is the level of the vertex in $C$. Note that the structure of the labels and the label arrangement is the same as described in Section \ref{subsec:stage1}. Thus, the proof of Lemma \ref{lem:isDLT} applies here for the new connected component and its labeling. Therefore, $C$ is a DLT. This completes the description of connection stages. Their steps are summarized in Algorithm \ref{alg:connection}. In what follows, we analyze the algorithm.

\begin{algorithm}[H]  
\caption{Connection($G$)}  \label{alg:connection}
\begin{algorithmic}[1]

\STATE /********  First Connection Stage  ***********/

\FOR {each tree $T \in G$ in parallel}
    \STATE Scan all edges $(u,w)\in G$, such that $u \in T$ and $w \notin T$. Search for an edge $e = (u,w)$, such that the first coordinate in $L(w)$ is the smallest, and this coordinate is 
    smaller than the ID of $T$.
    \IF{such an edge $e$ was found}
        \STATE Connect $T$ to the tree to which $w$ belongs, and make that tree a parent of $T$.
        \STATE The vertex $u$ becomes the new root of $T$. Set the label of each vertex $v \in T$ to be $\langle L(T), l    \rangle$ where $l = dist_T(u,v)$.
    \ELSE 
        \STATE Mark $T$ as a local minimum tree.
    \ENDIF
\ENDFOR

\STATE /********  Second Connection Stage  ***********/

\FOR {each local minimum tree $T'_i$ in parallel}
    \STATE Perform a convergecast in $T'_i$ to check if $T'_i$ became a parent of another tree.
    \IF{$T'_i$ did not became a parent of another tree}
        \STATE Choose an edge $e'$ with one endpoint in $T'_i$, connecting to another connected component arbitrarily.
    \ENDIF
\ENDFOR

\STATE Wake up all vertices in $G$ for one round and exchange knowledge about all edges $e'$ that cross between components $T'_i$, $T'_j$, $i \neq j$. This is achieved by sending component labels over the edges $e'$ in parallel within one round. In the end of this round, all vertices enter the sleep state.

\FOR {each connected component $C$ in parallel}
    \STATE Using convergecast collect information about all vertices of $C$ in the root of $C$.
    \STATE The root $r_C$ of $C$ executes a BFS internally and reassigns labels to each vertex in $v \in C$, such that the new label is $\langle L(C), l    \rangle$, where $l = dist_C(u,v)$.
    In this reassignment, the root $r_c$ remains the same.
    \STATE Broadcast the result to all vertices in $C$.
\ENDFOR

\end{algorithmic}
\end{algorithm}

\subsection{Analysis}

We turn to analyse our algorithm for spanning a DLT on the input graph $G$. We prove several lemmas and conclude with the main result for our scheme.

\begin{lem}
Let $C$ be a connected component at the end of a connection phase. Then $C$ is a DLT.
\end{lem}
\begin{proof}
Since our connection phase is divided into two stages, we also divide our proof to two. First, Let $C' \in C$ be a connected component at the end of stage one of the connection phase. Let $S$ be the set of DLTs composing $C'$. We define a graph $H$, where the vertex set is a DLT $T_i \in S$ that corresponds to $v_i$ in $H$. $H$ contains an edge $(v_i, v_j)$ if there is an edge between the respective DLTs in $S$. We assign $ID(v_i) = L(T_i)$, that is each vertex has Id equal to the label of the root of its corresponding DLT in $C'$. Let $T_i$ and $T_j$ be two DLTs in $C'$ such that $T_i$ oriented an edge towards $T_j$. Therefore, $L(T_i) > L(T_j)$. Then we have that $v_i$ is a child of $v_j$ in $H$ and $ID(v_i) > ID(v_j)$. This gives us an acyclic orientation such that each vertex has a single outgoing edge (towards its parent). Thus, $H$ is an oriented tree. Not only that, $H$ is a DLT since each child has an ID greater than its parent. Then we look at the edge $(w, u)$ which is the edge that originally connects $T_i$ to $T_j$ in $C'$. We make sure that, internally, each vertex $v_i \in H$ assigns the IDs of its corresponding DLT $T_i$ to root $T_i$ in $w$ while preserving the DLT property in this new assignment. Therefore, $H$ is a DLT where each vertex is by itself a DLT. Thus, $C'$ is a DLT. \\
Let $R = {R_1, R_2,\dots}$ be the set of DLTs that connect to $C'$ in stage two of the connection phase. We again look at each DLT as a vertex. First, since each such vertex has ID which is a local minimum, there are no two such vertices which are neighbors, and therefore there are no two vertices that can choose each other. Furthermore, we make sure that these vertices are not connected to any other vertex in phase one of the connection phase, thus they are an independent set. $C'$ is a DLT and each DLT in $R$ connects to $C'$ through a single edge $(a,b)$ that is oriented towards $C'$. We again have an acyclic orientation where each vertex has a single parent. Thus, $C$ is a tree. Let us look at the DLTs as vertices again. Then in our algorithm we set $L(v_{R_i}) > L(v_{C'})$. This guarantees that the induced graph on the vertex set $H \cup {v_{R_i}}$ is a DLT. Then, internally, we root $R_i$ in $a$ and reassign the labels of the vertices in $R_i$ to preserve its DLT property. Thus, the resulting component $C$ is a DLT of DLTs thus it is a DLT of itself.
\end{proof}

Next, we analyze the awake complexity of the algorithm. Our claim is that a connection phase halves the number of DLTs in the forest. This is quite straightforward. If a DLT is not connected in stage one of the connection phase, phase two considers it as problematic and makes sure it connects to another DLT. Thus, every DLT connects to another DLT and thus the number of DLTs is at least halved. \\
The next lemma analyses the performance of a connection phase.

\begin{lem}  \label{lem:conPhaseTime}
Each vertex $v$ is awake for at most $O(1)$ rounds in each connection phase.
\end{lem}
\begin{proof}
We show this step by step over the algorithm.
\begin{enumerate}
    \item {\em Each vertex sends the edge to the neighbor with the minimal label to the root of a DLT.} This is a convergecast procedure which we showed to take 2 awake rounds for each vertex.
    \item {\em The root chooses an edge and broadcasts it to all vertices in the DLT.} A broadcast also takes 2 awake rounds for each vertex.
    \item {\em Each DLT, internally, reassigns the labels of its vertices.} This can be done with one convergecast procedure and one broadcast procedure. This step takes up to four awake rounds.
    \item {\em A local minima DLT chooses a component to connect to and then receives the label of that component.} We wake up the entire graph for exactly one round. Each vertex adds a single awake round to its awake count.
    \item {\em A local minima DLT, internally, reassigns the labels of its vertices.} Again, this can be done using one convergecast procedure and one broadcast procedure. Each vertex is awake for at most 4 rounds.
    \item {\em All vertices in $G$ wake up for one round to learn the new labels of their 1-hop neighborhood.} Another single awake round is added for each vertex awake count.
\end{enumerate}
Overall, each vertex is awake for $O(1)$ rounds in each connection phase.
\end{proof}

\noindent We can now conclude the analysis of our DLT spanning algorithm. Since the number of DLTs is at least halved in each phase, there are $O(\log n)$ such phases, and the following theorem is directly obtained from Lemma \ref{lem:conPhaseTime}.

\begin{thm}  \label{thm:DLTrunningTime}
A DLT for any input graph $G$ can be deterministically computed in $O(\log n)$ awake rounds in the sleeping model. 
\end{thm}

As shown in Lemma \ref{lem:DLTprop}, the action of convergecast and the action of broadcast on the resulting spanning tree each require only $O(1)$ time in the sleeping model and thus we can collect the topology of the entire input graph to the root of our spanning tree, calculate the solution to the problem $\mathcal{P}$ deterministically and broadcast the solution to all vertices through our spanning tree, again in $O(1)$ time. This places an upper bound on the class of all decidable problems as we conclude in the following theorem.

\begin{thm}
Any decidable problem can be solved within $O(\log n)$ worst-case deterministic awake-complexity in the sleeping model.
\end{thm}

\section{A Tight Bound for DLT}  \label{sec:boundDLT}

In this section we prove that the complexity of DLT in the sleeping model is $\Omega(\log n)$. The proof is by a reduction from leader election on rings. For the latter problem, it is known that a certain number of messages must be sent in the network in order to solve it \cite{FL87}. In what follows we prove that this lower bound on messages implies a lower bound on awake rounds in the sleeping model. Before presenting the proof, we need the following lemma, which demonstrates a connection between the number of messages an algorithm produces and the complexity in the sleeping model.

\begin{lem} \label{lem:minMSGs}
Any algorithm $\mathcal{A}$ which requires at least $\Omega(\Delta n  \log n)$ messages for its execution has an awake complexity of at least $\Omega(\log n)$ in the sleeping model.
\end{lem}
\begin{proof}
Let the number of messages that must be sent during the execution of $\mathcal{A}$ be $c \Delta \cdot n \log n$ where $c > 0$ is a constant. We show that there is at least one vertex $v$ that must be awake for at least $c \log n$ rounds.  Assume for contradiction that all vertices in $G$ are awake for less than $c \log n$ rounds. 
Each vertex sends at most $\Delta$ messages (one across each adjacent edge) in a single awake round. If more than one message per edge per round is required, all these messages can be concatenated into a single message. Thus, each vertex sends less than $\Delta \cdot c \log n$ messages, and the overall number of messages in the execution is less than $n (\Delta \cdot c \log n) = c \Delta \cdot n \log n$.
This is a contradiction. Therefore, there must be a vertex that is awake for $\Omega(\log n)$ rounds. We conclude that the awake round complexity of $\mathcal{A}$ in the sleeping model is also $\Omega(\log n)$. Given that there are at least $\Omega(\Delta n \log n)$ messages and $n$ vertices and at most $\Delta n$ edges, on average, each vertex is awake for at least $c \log n$ rounds.
Thus, the running time of $\mathcal{A}$ (in the worst case and average case) is at least $\Omega(\log n)$.
\end{proof}

\noindent {\bf Remark:} An algorithm $\mathcal{A}$ that requires $\Omega(\Delta n \log n)$ messages has an awake complexity of $\Omega(\log n)$, not only in the worst vertex, but also on average over the vertices. (Such an average complexity is referred to as {\em vertex-averaged complexity} \cite{CGP20}.) Indeed , if the vertex-averaged awake complexity is $o(\log n)$ then the sum of awake rounds for all vertices is $o(n \log n)$, and the number of messages is $o(\Delta n \log n)$, according to the proof of lemma \ref{lem:minMSGs}. 

Next, we employ Lemma \ref{lem:minMSGs} in order to prove that DLT requires $\Omega(\log n)$ complexity in the sleeping model. We show this for a ring graph by a reduction from the leader election problem.

\begin{thm}    \label{thm:omegaIDassign}
Let $t > 0$ be an arbitrarily large integer, and $\mathcal{A}$ any deterministic algorithm  for the DLT problem, which requires $t$ rounds in the $\mathcal{LOCAL}$ model. Then there is an ID assignment from a sufficiently large range, as a function of $t$, such that $\mathcal{A}$ requires $\Omega(\log n)$ awake-complexity in the sleeping model.
\end{thm}
\begin{proof}
The proof is by contradiction. 
Assume that there is an algorithm $\mathcal{A}$ with awake-round complexity of $o(\log n)$, overall complexity $t > 0$, for ID assignment from an arbitrarily large range. Then $\mathcal{A}$ uses at most $o(n \Delta \log n)$ messages (see Lemma \ref{lem:minMSGs}). Let $C$ be an $n$-vertex cycle graph. The maximum degree of $C$ is $\Delta = 2$. We execute $\mathcal{A}$ on $C$ in the ordinary (not-sleeping) $\mathcal{LOCAL}$ model. We obtain a DLT of $C$ within $t$ rounds. Now, the root can be elected as the leader, and the other vertices know that they are not the root. In a DLT they also know the ID of the root. Thus, we have an algorithm for leader election in the $\mathcal{LOCAL}$ model which employs at most $o(n \log n)$ messages. 

According to \cite{FL87}, the leader election problem requires $\Omega(n \log n)$ messages, if vertex IDs are chosen from a set of sufficiently large size $R(n,t)$, where $R$ is the Ramsey function and $t$ is the running time of the algorithm. This is a contradiction.

\end{proof}

\noindent It follows that any problem whose solution can be used to elect a leader within $o(\log n)$ additional awake rounds requires $\Omega(\log n)$ awake-complexity. We denote the class of such problems by {\bf DLT-hard} problems. Theorems \ref{thm:DLTrunningTime} and \ref{thm:omegaIDassign} directly give rise to the following corollary.

\begin{thm}  \label{thm:endSec3}
The class of DLT-hard problems has a deterministic complexity tight bound of $\Theta(\log n)$ in the sleeping model.
\end{thm}

\noindent {\bf Remark:} Note that in Section \ref{sec:conPhase} we showed that the DLT problem is complete in the class of decidable problems and Theorem \ref{thm:endSec3} states it is $\Omega(\log n)$-hard under this class.

\section{Solving Oriented-Local Problems}   \label{sec:olocal}

In this section we devise an algorithm for solving a class of {\em Oriented-Local} problems. This class contains all problems which, given an acyclic orientation on the edge set of the graph, can be solved as follows. Each vertex awaits all neighbors on outgoing edges to produce an output, and then computes its own output as a function of the outputs of these neighbors. (Vertices with no outgoing edges produce an output immediately.) We define this class formally.

\begin{dfn}  \label{def:1OLOCAL}
The class of {\bf 1-hop Oriented Local Problems ({\bf 1-O-LOCAL})} consists of all problems that, given an acyclic orientation $\mu$ on the edge set of $G$, can be solved in the following way. Let $v$ be a vertex in $G$. Let $U$ be the set of neighbors of $v$ in its 1-hop neighborhood which precedes $v$ in the orientation $\mu$, i.e., the vertices connected by outgoing edges from $v$. Let $s(p)$, for $p \in U$, be the solution of the problem. Then, $v$ can internally calculate $s(v)$ with the knowledge of $\{s(p)\ |\ p \in U \}$. \\
The class of {\bf Oriented Local Problems ({\bf O-LOCAL})} is a generalization of 1-O-LOCAL, where the set $U$ contains all vertices on paths that emanate from $v$, rather than $v$'s immediate neighbors on such paths.
\end{dfn}

\noindent As one can tell, a solution for a problem in this class depends on a given orientation. Such orientation can be calculated or given as an input to the algorithm. In this work we assume that no orientation is given and we are forced to calculate one as part of the solution. We note that MIS and $(\Delta+1)$-vertex-coloring are examples of well-studied problems which fall in the class of {\bf 1-O-LOCAL} problems. \\

We start with an algorithm for $O(\Delta^2)$-vertex-coloring in $O(\log^*n)$ time \cite{L87}. This gives us an orientation of the edges where we orient edges in descending order, i.e., each edge is oriented towards the endpoint of a smaller color. We have all vertices in $G$ awake during the entire coloring algorithm. Let $q = O(\Delta^2)$ be an upper bound on the number of colors of the algorithm of Linial such that $q$ is a power of 2. At the next stage each vertex $v$ builds a binary search tree internally. The size of the tree is $2q-1$. The root of the tree receives the label in the middle of the range $[2q-1]$, which is $q$. Now we have $q-1$ values in each side of the tree. Specifically, $\{1,\dots,q-1\}$ for the left subtree and $\{q+1,\dots,2q-1\}$ for the right subtree. We choose the middle of the range $\{1,\dots,q-1\}$ for the left child of the root and the middle of the range $\{q+1,\dots,2q-1\}$ for the right child of the root. We continue this recursively, so that each node of the tree obtains a unique value from $[2q - 1]$. 

Now we recolor the vertices of the input graph using the following mapping. The recoloring is performed by all vertices in parallel, with no communication whatsoever. We map the elements from $[q]$ to the set of values appearing in the leaves of the binary tree. The mapping is the same in all vertices. Specifically, for each $i \in [q]$, the $i$-th element is mapped to the label of the $i$-th left-most leaf of the tree. Consequently, all vertices that were initially colored by the color $i$ switch their color to the label of the $i$the leaf in the tree. Note that each pair of neighbors select distinct leaves of the tree, since their original colors are distinct. Therefore, the coloring after the mapping is proper as well. 

Next, we switch to the sleeping state for all the vertices in the graph, and start solving a given {\bf 1-O-LOCAL} problem $P$. For the sake of simplicity, we proceed with the problem of MIS, but our method can be applied to any {\bf 1-O-LOCAL} problem, as would be obvious from the description of the algorithm. The scheme is as follows. Each vertex $v$ employs its color in the $O(\Delta^2)$-coloring, and a respective leaf in the binary tree, whose value equals the color of $v$. Let $R$ be the path from the leaf of the color of $v$ to the root of the binary tree. Let $r(v) = \{r_1,\dots,r_k\}$ be the values appearing in $R$. We denote $r' = r_1$. Note that some values in $r(v)$ may be greater than $r'$, while other values may be smaller than $r'$.  Then $v$ awakes at each round $r_i \in r(v)$, and sends a message to its awake neighbors about its state, e.g., whether it is in the MIS, not in the MIS or undecided. It also receives such messages from its awake neighbors in these rounds. Recall that $r' = r_1$ is the round number that is equal to the color of $v$. In round $r'$  the vertex  $v$  makes a decision if to join the MIS or not according to the information received from outgoing edges. The neighbors on such edges have smaller colors, and thus have made a decision before round $r'$. Indeed, the following lemma proves that $v$ has all the information from vertices of lower colors when round $r'$ arrives.

\begin{lem}
At round $r'$, which is mapped to the color of $v$, all vertices with colors smaller than that of $v$ have already made a decision. Furthermore, their decisions have been passed to $v$ in a previous round.  
\end{lem}
\begin{proof}
We prove the lemma by induction on the colors of the orientation. \\
{\bf Base:} For the left-most leaf in the tree, the mapping of the first color in the orientation maps to the first awakening round of the algorithm. Vertices with the first colors of the orientation have no outgoing edges and need not wait for decisions of any of their neighbors. As they wake at the first round they make a decision to be in the MIS and sleep again.\\
{\bf Step:} Let $v$ be a vertex which awakes in round $r'$ and assume by induction that all neighbors with lower colors already made a decision to be in the MIS or not.
Let $\hat \Delta$ be the number of neighbors of $v$ with colors smaller than the color of $v$. Let $S = s_1,\dots,s_{\hat \Delta}$ be the rounds mapped to each color of these neighbors.
Then we have $r' > \{s_i \in S\}$. Thus, in the binary tree, for each $i \in [\hat \Delta]$, $r'$ and $s_i$ have a lowest common ancestor with ID $t$, such that $s_i \ < t < r'$. (See Figure 1.) This is because a lowest common ancestor of two leaves must have these leaves in distinct subtrees rooted in its children. Otherwise, if both leaves belong to the same subtree of a child of a common ancestor, it is not the lowest one. 

Let $u$ be a neighbor of $v$ with a color corresponding to the mapping $s_i$. We note that $s_i$ must be in the subtree rooted in the left child of the ancestor of ID $t$ and $r'$ is in the subtree rooted in the right-child of the ancestor $t$. Both $u$ and $v$ are awakened in round $t$ according to our algorithm. At round $t$, since $s_i < t$, the vertex $u$ already made a decision if it is in the MIS or not, by the induction hypothesis. Thus, $u$  sends a message with its decision to $v$ at round $t$. Since $r' > t$, at round $t$, $v$ simply receives the messages and awaits round $r'$ to make a decision. (During this waiting period, the vertex $v$ may communicate with additional neighbors.)
When round $r'$ finally arrives, all neighbors with lower colors, those in $S$, have made decisions and sent their decision in the round corresponding to some common ancestor with $v$ in the binary tree. Thus, $v$ has learnt the decisions of neighbors with smaller colors than its own. Finally, $v$ makes a decision in round $r'$ according to all the decisions made by neighbors in $S$.
This concludes the proof of the lemma. 
\end{proof}

\begin{figure}[h]
    \centering
    \includegraphics[height=3.375in, width=6in]{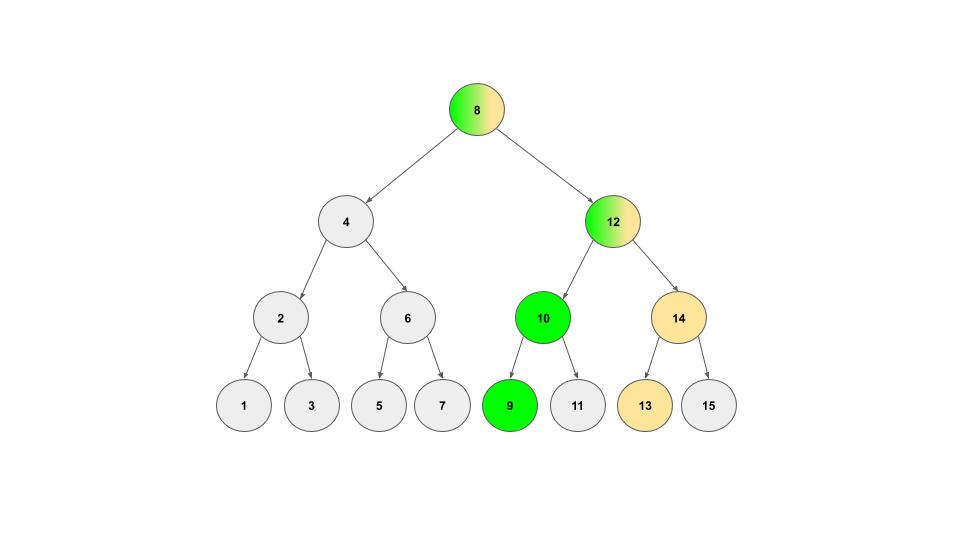}
    \caption{An example of a tree in the internal memory of each processor, for $q = 8$. A pair of neighbors $u,v$ (not depicted) are colored by 9 and 13, respectively. In green are the rounds in which vertices that correspond to  color 9 are awake. In orange are the rounds in which the vertices that correspond to color 13 are awake. The lowest common ancestor of these two colors is 12. In this round both $u$ and $v$ awake, and $v$ receives the decision of $u$. (Note that both $u$ and $v$ are also awake in round 8, but in this round $u$ may have not reached a decision yet, since its color is $9 > 8$.)}
    \label{fig:treeDemon}
\end{figure}

For a problem $P$ in {\bf 1-O-LOCAL}, a vertex $v$ can make its decision in round $r'$. For example, the following decisions are made in some  well-studied {\bf 1-O-LOCAL} problems: For MIS, $v$ joins the MIS if all neighbors with lower colors are not in the MIS. For $(\Delta+1)$-vertex-coloring, $v$ chooses a new color from the palette $[\Delta+1]$ which is not yet chosen as a new color by its neighbors with lower old-colors (i.e., colors according to the initial orientation). 

The depth of a binary tree with $O(\Delta^2)$ leaves is at most $O(\log \Delta)$. Thus, the size of a path $R$ from a leaf to the root is at most $O(\log \Delta)$. The vertex $v$ only awakens in rounds corresponding to keys appearing along $R$, and thus $v$ awakens in at most $O(\log \Delta)$ rounds. This provides us with the complexity of our algorithm in the following theorem.

\begin{thm}
Any {\bf 1-O-LOCAL} problem can be solved in $O(\log \Delta + \log^*n)$ deterministic awake-complexity in the sleeping model.
\end{thm}

Note that each vertex is able to accumulate all information received from outgoing neighbors and pass it later to incoming neighbors, when these neighbors ask it for its output. Consequently, each vertex learns the information from all vertices that emanate from it in the orientation. Thus, each vertex is able to produce an output not only as a function of its outgoing-neighbors outputs, but as a function of all output of vertices that emanate from it. In other words, any {\bf O-LOCAL} problem can be solved this way. Hence, we obtain the following corollary.

\begin{col}
Any {\bf O-LOCAL} problem can be solved in $O(\log \Delta + \log^*n)$ deterministic awake-complexity in the sleeping model.
\end{col}

\section{Lower Number of Clock Rounds}

In the $\mathcal{LOCAL}$ model any decidable problem can be solved in $O(Diam(G))$ clock rounds by spanning a tree from the vertex with the smallest ID and using broadcast and convergecast actions. This is a known fact. Also, some decidable problems are known to have a lower bound of $\Omega(n)$ rounds, for example 2-coloring of a path \cite{L87}. One may see that the method we devised in the DLT is beneficial in the sleeping model, but notice that each connection phase requires activation of vertices over $O(n^2)$ clock rounds.  This is because each vertex waits for the round which equals to its ID, which is a vector of size 2, and each of its coordinates is  of range $O(n)$. 

In this section we show that a connection phase can be done more efficiently in respect to the number of clock rounds. This shows that our method does not require the entire graph to wait for a large number of clock rounds to complete when one observes all the vertices as a whole, as one would do in the $\mathcal{LOCAL}$ model. This comes at an additional small price of awake times of vertices. More specifically, we show that each connection phase takes $O(n \log^* n)$ clock rounds while each vertex is awake for $O(\log^* n)$ rounds during each phase. This is in contrast to our original method in which the connection phase takes $O(n^2)$ clock rounds while each vertex is awake for $O(1)$ rounds in each phase. 

The main idea of this algorithm is to form trees of a bounded height in each connection phase. This is in contrast to our original algorithm, where a tree obtained from merging several DLTs may be arbitrarily high, which increases the waiting time for each vertex until it awakes. In the current version, however, we partition the set of trees that would like to join a common component into smaller components with bounded height. This limits the waiting time for a vertex until it needs to awake. Nevertheless, each such component still contains at least two DLTs of the previous phase. Thus, the number of DLTs at least halves with each phase, and the algorithm completes within $O(\log n)$ connection phases. Moreover, when each tree of the previous phase is considered as a vertex in a component of the current phase, the height of the component is guaranteed to be bounded by a constant. Consequently, the merging of trees in the component requires just $O(n \log^*n)$ clock rounds, rather than $O(n^2)$ as in our original algorithm. In order to partition a component of large height into components of bounded height we employ a 3-coloring algorithm of the component tree, and then perform partitioning and merging using these colors. 

We start by assuming that at the beginning of each connection phase, in each connected component $C$, for each vertex $v$, the ID vector is composed as $<ID(r_C), l>$, where $r_C$ is the root of $C$ (we remind that $C$ is a DLT at the start of a connection phase) and $l$ is the distance of $v$ from $r_C$. 
We will end our connection phase with each newly formed DLT having this property. Then, indeed, a broadcast and convergecast actions in such a DLT takes $O(n)$ clock rounds. For a broadcast, each vertex awakes in rounds $l-1$ and $l$. For convergecast each vertex awakes in rounds $n-l$ and $n-l-1$. The above property is trivial at the beginning of our spanning algorithm as each vertex $v \in G$ is its own DLT. It simply constructs its ID as $<ID(v), 0>$. We now describe a single connection phase, again, starting with the above property assumed and ending it with the above property achieved. 

We focus on stage 1 of the connection phase as described in Subsection \ref{subsec:stage1}. Stage 2 remains unchanged. We start with each existing DLT choosing a parent DLT as described in that subsection. The result is a new connected component $T$ which is a tree of connected trees. That is, $T$ can be viewed as a directed tree of a set of vertices $V(T)$, such that each vertex in $V(T)$ is a DLT. Let us look at the graph $H$ where we refer to each directed tree in $V(T)$ as a vertex. That is $H = (V', E')$ where there is a vertex in $V'$ for a corresponding DLT in $T$. The ID of a vertex in $H$ equals the label of its corresponding DLT. There is an edge $(v,u) \in E'$ if $v, u$ correspond to connected DLTs in $T$. The edge $(v, u)$ is directed in the same direction as the edge connecting the corresponding DLTs in $T$. Thus, $H$ is a directed tree just as $T$ is a directed tree of DLTs. Each ID in $H$ is of the range $O(n^2)$. We give some notation to simplify the discussion ahead. We denote $T_i$ each of the DLTs which make the component $T$. For each $T_i$, we denote $v(T_i)$ the corresponding vertex in $H$. Vice versa, we denote $T(v)$ the DLT in $T$ which corresponds to a vertex $v \in H$. 

We 3-color the directed tree $H$ using the algorithm of Goldberg, Plotkin and Shannon \cite{GPS88}. During this algorithm, after each step of the coloring algorithm, we pause such that each $T(v)$ can convergecast and broadcast internally. This way, all vertices in $T(v)$ have the knowledge that $v$ should have for the continuation of the coloring algorithm in $H$. Since our assumption allows each DLT $T(v)$ to perform the broadcast and convergecast actions in $O(n)$ clock rounds, the overall number of clock rounds required to finish coloring $H$ is $O(n \log^*n)$. Each vertex in $G$ has to be awake for at most $O(1)$ rounds for spreading information between each step in the coloring algorithm. So each vertex is awake for at most $O(\log^*n)$ rounds during the coloring of $H$. 

Now our goal is to break $H$ into a forest of trees each of depth at most 3. The steps of this procedure are described below. This procedure allows us to obtain the desired structure within a sufficiently small number of clock rounds. To keep track of the number of clock rounds that pass and the number of awake rounds for each vertex, we will analyze these complexity measures in each step of the following procedure that consists of four steps. 

\begin{enumerate}
    \item Each vertex $v$ of the color 1 chooses its parent $u$ in $H$ as a parent and sends a message to $u$ that it was chosen. For this communication we awake all vertices in $G$ for a single awakening round so that both endpoints of the edge connecting $T(v)$ and $T(u)$ can communicate. $v$ and $u$ mark themselves as "connected". (In $T$ this means a convergecast and broadcast actions inside $T(v)$ and $T(u)$). This takes at most $O(n)$ clock rounds and at most $O(1)$ awake rounds for each vertex.
    
    \item Each vertex $v$ of the color 2 which is not yet marked connected, and that has a parent that is not marked "connected", chooses its parent in $H$ as a parent. Denote that parent as $u$. Again, $v$ notifies $u$ that it was chosen and both are marked "connected". This takes at most $O(n)$ clock rounds and at most $O(1)$ awake rounds for each vertex.
    
    \item Repeat the same as in step 2 for the color 3.
    
    \item Let $v$ be a vertex in $H$ which is not marked connected. Then $v$ connects to its parent in $H$.
\end{enumerate}

 We note  that the above process takes at most $O(n)$ clock rounds to finish and adds $O(1)$ awake rounds to each vertex in $G$. So far, the total number of clock rounds is at most $O(n \log^*n)$ and the awake time is at most $O(\log^*n)$.
Let us denote the forest created as $H_1, \dots, H_k$. We prove the following lemma.

\begin{lem}
Each $H_i$ is of depth at most 3.
\end{lem}
\begin{proof}
For any $H_i$ to be of depth greater than 3, there must be a vertex in $H_i$ with a grand-grandchild. Let $v,u,w,t$ be the vertices where $v$ is the child of $u$, $u$ is the child of $w$ and $w$ is the child of $t$ (thus, $v$ is a grand-grandchild of $t$). Let us assume that $v$ chose $u$ before Step 4. We show that $v$ cannot even be a grandchild of $w$ let alone a grand-grandchild of $t$. Since $v$ chose $u$ before step 4, $u$ itself could not have been marked "connected" at that point. Also, since the colors of $v$ and $u$ are different in the 3-coloring of $H$, $u$ did not choose $w$ as a parent simultaneously when $v$ chose $u$ as a parent. Thus, when it came the turn for $u$ to act, $u$ is already marked "connected" and cannot have chosen $w$ as a parent. Nor would it have done so in step 4 since only unconnected vertices act in that step. Especially, $v$ cannot be the grand-grandchild of $t$. \\
According to what we have just shown, until step 4, all trees are of depth at most 2. The set of unconnected vertices in step 4 is independent and can only connect to a tree constructed in steps 1-3. Thus, the depth of each such tree can grow by at most 1 in step 4 and reach the depth of at most 3 as required to prove.
\end{proof}

\noindent Note that since in step 4 we guarantee that unconnected vertices become connected, it is clear that the depth of each $H_i$ is at least 2 and we at least halve the number of connected components from the start of the connection phase. 

Next, we return to considering $C$, which is now partitioned into subtrees of DLTs. a DLT $T_i$ remains connected to its parent $T_j$ only if there is a tree $H_i$ in which $v(T_i)$ is a child of $V(T_j)$. We thus partitioned $C$ into smaller components. Let $H_i$ be such a tree. Each DLT $T_i$ can be in one of three states: Either $v(T_i)$ is the root in $H_i$ and we denote $T_i$ as $T_{root}$; or $v(T_i)$ is a leaf and we denote it as $T_{leaf}$; or $v(T_i)$ is exactly between a root and a leaf in $H_i$, in which case we denote it as $T_{middle}$.

\begin{figure}
    \centering
    \includegraphics[height=3.375in, width=4in]{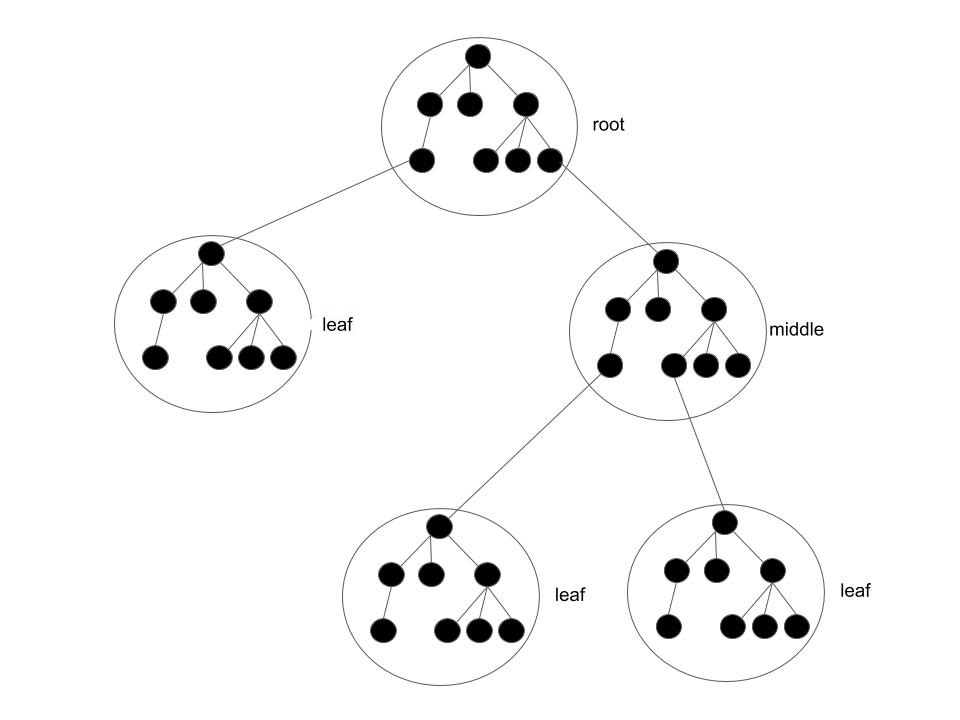}
    \caption{DLTs labeling showing a 3-depth connection.}
    \label{fig:DLTs}
\end{figure}

\noindent We can now finally change the IDs in each $T_i$ to achieve the property with which we began the connection phase. W.L.O.G. we assume that $H_i$ is of depth 3. The process is as follows:

\begin{enumerate}
    \item $T_{root}$ remains unchanged. We wake up the vertices of $T_{root}$ and all vertices in the DLTs marked $T_{middle}$ for one round. $T_{root}$ sends the ID of $r_{T_{root}}$ and the height of $(T_{root})$ to the root of each $T_{middle}$. We remind that all vertices of $T_{root}$ have this knowledge internally. All vertices sleep again.
    
    \item In each $T_{middle}$ the root calculates internally new IDs for the vertices in $T_{middle}$ of the form $<ID(r_{T_{root}}), l>$ where $l$ is the distance of the vertex from $r_{T_{root}}$ which can be calculated because the root knows its distance from $r_{T_{root}}$. We perform a broadcast in $T_{middle}$ in $O(n)$ clock rounds and $O(1)$ awake rounds. All vertices sleep again. Note that now all vertices in $T_{middle}$ know the ID of the $r_{T_{root}}$ and their distance from it.
    
    \item We repeat the above two steps where $T_{middle}$ transmits to all its children, those marked as $T_{leaf}$ and the roots of these leaves reassign labels internally in their respective trees (similar to what is described in step 2). 
\end{enumerate}

\noindent Overall, the above process takes another $O(n)$ clock rounds and additional $O(1)$ awake rounds for each vertex in $G$. When the process terminates, each label of each vertex in our newly formed DLT, denote $L$, is of the form $<ID(r_{L}), l>$ where $l$ is the distance of the vertex from $r_{L}$, which is exactly the property we wanted to preserve and which allows us to move on to the next connection phase.

The number of connection phases needed remains $O(\log n)$, since we guarantee that each DLT is connected to at least one other DLT in each phase. Each connection phase performs the coloring process in which each vertex is awake for at most $O(\log^*n)$ stages, each can take up to $O(n)$ clock rounds due to the dissemination of knowledge inside each connected component that functions as a vertex in $H$. Thus, the total number of clock rounds is $O(n \log n \log^*n)$ instead of $O(n^2 \log n)$ for finding a DLT for the input graph $G$. The awake complexity of the sleeping model is $O(\log n \log^* n)$ instead of $O(\log n)$, since each dissemination requires each vertex in a connected component for only $O(1)$ awake rounds as shown in Lemma \ref{lem:DLTprop}. Thus, we improved the number of clock rounds significantly while paying a small price in awake complexity. The final result is a DLT for $G$ where each vertex has knowledge of its distance from the root. Thus, we can perform the broadcast and convergecast procedures using only $O(n)$ clock rounds which is optimal.

\begin{thm}
There is a deterministic algorithm for the DLT problem with $O(\log n \log^*n)$ worst-case awake time which terminates within $O(n \log n \log^*n)$ clock rounds.
\end{thm}

\section{Sleeping in the $\mathcal{CONGEST}$ Model}

In this section we show that our construction of a DLT on $G$ can be achieved also in the $\mathcal{CONGEST}$ model. This is of great importance to some well-studied problems in which the exchanged information is of restricted size. We define a class of problems {\bf C-CONGEST} as follows.

\begin{dfn}  \label{def:cCongest}
The class of problems {\bf C-CONGEST}, or the {\em Congested Combinations Class}, is the class of problems which have solutions with the following properties:
\begin{enumerate}
    \item The solution can be expressed by using up to $O(\log n)$ bits.
    \item Given two solutions on two subgraphs $R_1$ and $R_2$, one can compute the solution on the graph $R_1 \cup R_2$ using some sequential algorithm without incurring further communication.
\end{enumerate}
\end{dfn}

\noindent Note that the fact that the solutions can be expressed as in the definition does not mean the problem is easy to compute in the distributed setting. It may well be possible that the problem has a lower bound of $ \omega(\log n)$ and thus has a solution less efficient than constructing a DLT. Examples of problems in {\bf C-CONGEST} are leader election, computing exact number of edges and average degree.

We thus need to show that such a DLT is possible to compute in the $\mathcal{CONGEST}$ model. We will do so by going over the steps of our DLT algorithm as we did in Lemma \ref{lem:conPhaseTime} and show that each step can be performed in a congested network. As we will see, the steps that need the most attention are the steps where a DLT reassigns the labels internally to re-orient the edges (steps 3 and 5 in Lemma \ref{lem:conPhaseTime}). To this end, we devise a sub-procedure that allows this reassignment to occur in the sleeping model, while using messages of size at most $O(\log n)$ bits. We remind that the assignment of the labels is done by defining a new root and re-orienting the edges of the DLT accordingly. For a DLT $T$, The labels are of the form of $\langle L(T), l \rangle$, where $l$ is the distance of the vertex from the new root of $T$. \\

\noindent {\bf Reassigning Labels in the DLT in a Congested Network.} In the version suggested in Section \ref{sec:DLTmain} we aggregate all the knowledge to the root of the DLT, locally computing an assignment and broadcasting it to the vertices of the DLT. This requires messages of large size. Instead, we only aggregate the distance from the new root. This is done as follows. \\
Let $r$ be the root of the DLT $T$ and let $v$ be the new root from which we wish to reassign the labels. Let $u_1,\dots,u_k$ be the vertices in $T$ on the path from $v$ to $r$ (going up the tree). We start by performing a broadcast where only the vertices $r, v, u_1, \dots, u_k$ take part. When it is time for $v$ to be active, it sends a message to $u_1$, its parent, that it is the new root. Thus, when it is time for $u_1$ to send a message to its parent, it sends a message to $u_2$ that it is in distance 1 from the new root. We continue this where each $u_{i}$ sends its own distance from $v$ to $u_{i+1}$. Note that the distance is bounded by $n$, thus we use only $O(\log n)$ bits for each message. See Fig. \ref{fig:distOnPath} for an example. This discussion is summarized in the next lemma.

\begin{lem}   \label{lem:boundMSGsizeDLT}
The reassignment of labels in a connection phase in a DLT can be done using messages of size $O(\log n)$ bits.
\end{lem}

\begin{figure}
    \centering
    \includegraphics[height=3.375in, width=4.7in]{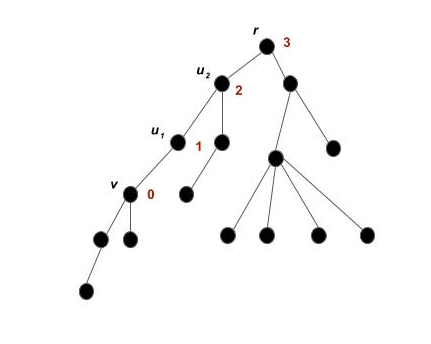}
    \caption{Distance from v along the path to the root}
    \label{fig:distOnPath}
\end{figure}

Now we move on to perform a propagation of the new label assignment in $T$. Each vertex in its turn simply sends its distance from $v$ to all of its children. Note that at the start of this stage, each vertex in the set $r, v, u_1, \dots, u_k$ knows its distance from $v$ and uses that distance in the message. Since, again, we send distances that are bound by $n$ we only use $O(\log n)$ bits for each message. See Figure \ref{fig:distFromV}. We conclude this section with the following theorem.

\begin{figure}
    \centering
    \includegraphics[height=3.375in, width=5.5in]{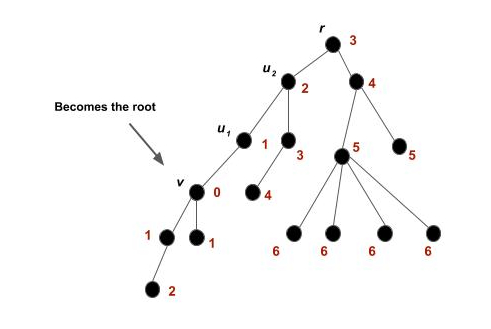}
    \caption{Each vertex knows its distance from v.}
    \label{fig:distFromV}
\end{figure}

\begin{thm} \label{thm:congestDLT}
The construction of a DLT can be done in a congested network in $O(\log n)$ deterministic worst-case awake time in the sleeping model.
\end{thm}

\begin{proof}
We show that each step of the algorithm can be executed using messages of $O(\log n)$ bits.
\begin{enumerate}
    \item {\em Each vertex sends the edge to the neighbor with the minimal label to the root of a DLT.} Each vertex passes to its parent the edges with the minimal label seen so far in the subtree including the edge it itself chose. Since each vertex aggregates exactly one edge, we communicate only $O(\log n)$ bits across each edge in the DLT.
    \item {\em The root chooses an edge and broadcasts it to all vertices in the DLT.} Again, each vertex propagates a single edge to its children in the DLT. We communicate only $O(\log n)$ bits across each edge in the DLT.
    \item {\em Each DLT, internally, reassigns the labels of its vertices.} According to Lemma \ref{lem:boundMSGsizeDLT} this requires messages of size at most $O(\log n)$ bits.
    \item {\em A local minima DLT chooses a component to connect to and then receives the label of that component.} The same as we did in step 1, this can also be done using only $O(\log n)$ bits across each edge.
    \item {\em A local minima DLT, internally, reassigns the labels of its vertices.} According to Lemma \ref{lem:boundMSGsizeDLT} this requires messages of size at most $O(\log n)$ bits.
    \item {\em All vertices in $G$ wake up for one round to learn the new labels of their 1-hop neighborhood.} Two labels are sent across each edge. This can be done using only $O(\log n)$ bits.
\end{enumerate}
\end{proof}

\section{Conclusion}

In this work we investigated the strength of Distributed Layered Trees in the sleeping model. We showed that the computation of such trees is complete and thus any decidable problem can be solved within $O(\log n)$ awake complexity. This raises the question of finding non-trivial sub-classes of decidable problems which one can solve in a more efficient way than using a DLT. We address this question by defining the {\bf O-LOCAL} class of problems and showing that it indeed can be solved more efficiently in the sleeping model. Since the $\mathcal{CONGEST}$ model is of great interest in the field of distributed networks, we investigated it as well, and obtained a class of problems that can be solved within logarithmic awake complexity by using only messages of logarithmic size. 
Another important aspect is the number of ordinary clock rounds of an algorithm with good awake complexity. While our simpler version of the algorithm has quadratic complexity of clock rounds, the more sophisticated variant gets closer to the optimal $\Theta(n)$ rounds.
Overall, we showed the strength of the sleeping model and the possibility of a significant energy conservation for distributed networks. \\

\noindent {\large \bf Acknowledgements}  The authors are grateful to the anonymous reviewers for helpful comments.



\bibliography{library}

\begin{thebibliography}{10}

\bibitem{BKO20}
A.~Balliu, F.~Kuhn, and D.~Olivetti.
\newblock Distributed edge coloring in time quasi-polylogarithmic in delta.
\newblock In {\em {PODC} '20: {ACM} Symposium on Principles of Distributed
  Computing, Virtual Event, Italy, August 3-7, 2020}, pages 289--298. {ACM},
  2020.

\bibitem{BDO14}
L.~Barenboim, S.~Dolev, and R.~Ostrovsky.
\newblock Deterministic and energy-optimal wireless synchronization.
\newblock {\em {ACM} Trans. Sens. Networks}, 11(1):13:1--13:25, 2014.

\bibitem{BE11}
L.~Barenboim and M.~Elkin.
\newblock Deterministic distributed vertex coloring in polylogarithmic time.
\newblock {\em J. {ACM}}, 58(5):23:1--23:25, 2011.

\bibitem{BEM17}
L.~Barenboim, M.~Elkin, and T.~Maimon.
\newblock Deterministic distributed (delta + o(delta))-edge-coloring, and
  vertex-coloring of graphs with bounded diversity.
\newblock In {\em Proceedings of the {ACM} Symposium on Principles of
  Distributed Computing, {PODC} 2017, Washington, DC, USA, July 25-27, 2017},
  pages 175--184. {ACM}, 2017.

\bibitem{BT19}
L.~Barenboim and Y.~Tzur.
\newblock Distributed symmetry-breaking with improved vertex-averaged
  complexity.
\newblock In {\em Proceedings of the 20th International Conference on
  Distributed Computing and Networking, {ICDCN} 2019, Bangalore, India, January
  04-07, 2019}, pages 31--40. {ACM}, 2019.

\bibitem{BKO12}
M.~Bradonjic, E.~Kohler, and R.~Ostrovsky.
\newblock Near-optimal radio use for wireless network synchronization.
\newblock {\em Theor. Comput. Sci.}, 453:14--28, 2012.

\bibitem{CDHHLP18}
Y.~Chang, V.~Dani, T.~P. Hayes, Q.~He, W.~Li, and S.~Pettie.
\newblock The energy complexity of broadcast.
\newblock In C.~Newport and I.~Keidar, editors, {\em Proceedings of the 2018
  {ACM} Symposium on Principles of Distributed Computing, {PODC} 2018, Egham,
  United Kingdom, July 23-27, 2018}, pages 95--104. {ACM}, 2018.

\bibitem{CGP20}
S.~Chatterjee, R.~Gmyr, and G.~Pandurangan.
\newblock Sleeping is efficient: {MIS} in \emph{O}(1)-rounds node-averaged
  awake complexity.
\newblock In {\em {PODC} '20: {ACM} Symposium on Principles of Distributed
  Computing, Virtual Event, Italy, August 3-7, 2020}, pages 99--108. {ACM},
  2020.

\bibitem{DHHV05}
J.~Deng, Y.~S. Han, W.~R. Heinzelman, and P.~K. Varshney.
\newblock Scheduling sleeping nodes in high density cluster-based sensor
  networks.
\newblock {\em Mob. Networks Appl.}, 10(6):825--835, 2005.

\bibitem{DHW98}
C.~Dwork, J.~Y. Halpern, and O.~Waarts.
\newblock Performing work efficiently in the presence of faults.
\newblock {\em {SIAM} J. Comput.}, 27(5):1457--1491, 1998.

\bibitem{FN01}
L.~M. Feeney and M.~Nilsson.
\newblock Investigating the energy consumption of a wireless network interface
  in an ad hoc networking environment.
\newblock In {\em Proceedings {IEEE} {INFOCOM} 2001, The Conference on Computer
  Communications, Twentieth Annual Joint Conference of the {IEEE} Computer and
  Communications Societies}, pages 1548--1557. {IEEE} Comptuer Society, 2001.

\bibitem{F17}
L.~Feuilloley.
\newblock How long it takes for an ordinary node with an ordinary {ID} to
  output?
\newblock In {\em Structural Information and Communication Complexity - 24th
  International Colloquium, {SIROCCO} 2017, Porquerolles, France, June 19-22,
  2017, Revised Selected Papers}, volume 10641 of {\em Lecture Notes in
  Computer Science}, pages 263--282. Springer, 2017.

\bibitem{F20}
L.~Feuilloley.
\newblock How long it takes for an ordinary node with an ordinary id to output?
\newblock {\em Theor. Comput. Sci.}, 811:42--55, 2020.

\bibitem{Fi20}
M.~Fischer.
\newblock Improved deterministic distributed matching via rounding.
\newblock {\em Distributed Comput.}, 33(3-4):279--291, 2020.

\bibitem{FGK17}
M.~Fischer, M.~Ghaffari, and F.~Kuhn.
\newblock Deterministic distributed edge-coloring via hypergraph maximal
  matching.
\newblock In {\em 58th {IEEE} Annual Symposium on Foundations of Computer
  Science, {FOCS} 2017, Berkeley, CA, USA, October 15-17, 2017}, pages
  180--191. {IEEE} Computer Society, 2017.

\bibitem{FL87}
G.~N. Frederickson and N.~A. Lynch.
\newblock Electing a leader in a synchronous ring.
\newblock {\em J. {ACM}}, 34(1):98--115, 1987.

\bibitem{GHS83}
R.~G. Gallager, P.~A. Humblet, and P.~M. Spira.
\newblock A distributed algorithm for minimum-weight spanning trees.
\newblock {\em {ACM} Trans. Program. Lang. Syst.}, 5(1):66--77, 1983.

\bibitem{GKKPS08}
L.~Gasieniec, E.~Kantor, D.~R. Kowalski, D.~Peleg, and C.~Su.
\newblock Time efficient k-shot broadcasting in known topology radio networks.
\newblock {\em Distributed Comput.}, 21(2):117--127, 2008.

\bibitem{GKM17}
M.~Ghaffari, F.~Kuhn, and Y.~Maus.
\newblock On the complexity of local distributed graph problems.
\newblock In {\em Proceedings of the 49th Annual {ACM} {SIGACT} Symposium on
  Theory of Computing, {STOC} 2017, Montreal, QC, Canada, June 19-23, 2017},
  pages 784--797. {ACM}, 2017.

\bibitem{GPS88}
A.~V. Goldberg, S.~A. Plotkin, and G.~E. Shannon.
\newblock Parallel symmetry-breaking in sparse graphs.
\newblock {\em {SIAM} J. Discret. Math.}, 1(4):434--446, 1988.

\bibitem{HKP98}
M.~Hanckowiak, M.~Karonski, and A.~Panconesi.
\newblock On the distributed complexity of computing maximal matchings.
\newblock In {\em Proceedings of the Ninth Annual {ACM-SIAM} Symposium on
  Discrete Algorithms, 25-27 January 1998, San Francisco, California, {USA}},
  pages 219--225. {ACM/SIAM}, 1998.

\bibitem{KPSY11}
V.~King, C.~A. Phillips, J.~Saia, and M.~Young.
\newblock Sleeping on the job: Energy-efficient and robust broadcast for radio
  networks.
\newblock {\em Algorithmica}, 61(3):518--554, 2011.

\bibitem{K20}
F.~Kuhn.
\newblock Faster deterministic distributed coloring through recursive list
  coloring.
\newblock In S.~Chawla, editor, {\em Proceedings of the 2020 {ACM-SIAM}
  Symposium on Discrete Algorithms, {SODA} 2020, Salt Lake City, UT, USA,
  January 5-8, 2020}, pages 1244--1259. {SIAM}, 2020.

\bibitem{L87}
N.~Linial.
\newblock Distributive graph algorithms-global solutions from local data.
\newblock In {\em 28th Annual Symposium on Foundations of Computer Science, Los
  Angeles, California, USA, 27-29 October 1987}, pages 331--335. {IEEE}
  Computer Society, 1987.

\bibitem{PXW09}
M.~Peng, Y.~Xiao, and P.~P. Wang.
\newblock Error analysis and kernel density approach of scheduling sleeping
  nodes in cluster-based wireless sensor networks.
\newblock {\em {IEEE} Trans. Veh. Technol.}, 58(9):5105--5114, 2009.

\bibitem{RG20}
V.~Rozhon and M.~Ghaffari.
\newblock Polylogarithmic-time deterministic network decomposition and
  distributed derandomization.
\newblock In {\em Proccedings of the 52nd Annual {ACM} {SIGACT} Symposium on
  Theory of Computing, {STOC} 2020, Chicago, IL, USA, June 22-26, 2020}, pages
  350--363. {ACM}, 2020.

\bibitem{S13}
J.~Suomela.
\newblock Survey of local algorithms.
\newblock {\em {ACM} Comput. Surv.}, 45(2):24:1--24:40, 2013.

\bibitem{TGAAA17}
C.~Titouna, A.~M. Gu{\'{e}}roui, M.~Aliouat, A.~A.~A. Ari, and A.~Adouane.
\newblock Distributed fault-tolerant algorithm for wireless sensor networks.
\newblock {\em Int. J. Commun. Networks Inf. Secur.}, 9(2), 2017.

\bibitem{WYHD19}
L.~Wang, J.~Yan, T.~Han, and D.~Deng.
\newblock On connectivity and energy efficiency for sleeping-schedule-based
  wireless sensor networks.
\newblock {\em Sensors}, 19(9):2126, 2019.

\bibitem{ZK05}
R.~Zheng and R.~Kravets.
\newblock On-demand power management for ad hoc networks.
\newblock {\em Ad Hoc Networks}, 3(1):51--68, 2005.

\end{thebibliography}

\bibliographystyle{abbrv}

\end{document}